\documentclass[12pt]{article}
\usepackage{geometry}
\usepackage{epsfig}
\usepackage[utf8x]{inputenc}
\usepackage{amsfonts}
\usepackage{amssymb}
\usepackage{amstext}
\usepackage{amsmath}
\usepackage{xspace}
\usepackage{theorem}
\usepackage{mathtools}
\usepackage{graphicx}
\usepackage{fullpage}
\usepackage{graphics}
\usepackage{colordvi}
\usepackage{color}
\usepackage{paralist}
\usepackage{bm}
\usepackage{mathpazo}
\usepackage{eepic}
\usepackage{boxedminipage}
\usepackage{theorem}

\usepackage{enumitem}
\usepackage{xspace}
\usepackage{hyperref}
\usepackage[capitalize]{cleveref}

\usepackage{bm}
\date{}

\newtheorem{theorem}{Theorem}
\newtheorem{lemma}[theorem]{Lemma}
\newtheorem{claim}[theorem]{Claim}

\newtheorem{observation}[theorem]{Observation}

\newtheorem{definition}{Definition}
\def\Box{\rule{2mm}{2mm}}

\newenvironment{proof}{\noindent {\it Proof.}}{\Box \vskip \belowdisplayskip}

\def\NP{\mbox{\bf NP}}

\newcommand{\abs}[1]{\mbox{$\left|#1\right|$}}

\newcommand{\set}[1]{\mbox{$\left\{#1\right\}$}}

\newcommand{\remove}[1]{}

\newenvironment{lp}[2]{\[\begin{array}{rcll}
                        \mbox{#1} & & #2 \\ 
                        \mbox{subject to}}{\end{array}\]}
\newcommand{\cnstr}[4]{\\ #1 & #2 & #3 & #4}

\title{Skeletons and Minimum Energy Scheduling}

\author{Antonios Antoniadis\thanks{University of Twente.}
\and Gunjan Kumar\thanks{National University of Singapore.} \and Nikhil Kumar \thanks{Hasso Plattner Institute Potsdam.}}

\begin{document}
\maketitle

\begin{abstract} 
Consider the problem where $n$ jobs, each with a release time, a deadline and a required processing time are to be feasibly scheduled in a single- or multi-processor setting so as to minimize the total energy consumption of the schedule. A processor has two available states: a \emph{sleep state} where no energy is consumed but also no processing can take place, and an \emph{active state} which consumes energy at a rate of one, and in which jobs can be processed. Transitioning from the active to the sleep does not incur any further energy cost, but transitioning from the sleep to the active state requires $q$ energy units. Jobs may be preempted and (in the multi-processor case) migrated.

The single-processor case of the problem is known to be solvable in polynomial time via an involved dynamic program, whereas the only known approximation algorithm for the multi-processor case attains an approximation factor of $3$ and is based on rounding the solution to a linear programming relaxation of the problem. In this work, we present efficient and combinatorial approximation algorithms for both the single- and the multi-processor setting. Before, only an algorithm based on linear programming was known for the multi-processor case. Our algorithms build upon the concept of a \emph{skeleton}, a basic (and not necessarily feasible) schedule that captures the fact that some processor(s) must be active at some time point during an interval. Finally, we further demonstrate the power of skeletons by providing an  $2$-approximation algorithm for the multiprocessor case, thus improving upon the recent breakthrough $3$-approximation result. Our algorithm is based on a novel rounding scheme of a linear-programming relaxation of the problem which incorporates skeletons.

\end{abstract}


\section{Introduction}

Energy consumption is one of the most important aspects of computing environments as supported, for example, by the fact that data centers already account for more than $1$\% of global electricity demand, and are forecast to reach $8$\% by $2030$~\cite{nature}. With that in mind, modern hardware increasingly incorporates power-management capabilities. However, in order to take full advantage of these capabilities algorithms in general, and scheduling algorithms in particular, must take energy consumption into consideration -- on top of the classical algorithm complexity measures of time and space.

In this work we study \emph{power-down mechanisms} which are one of the most popular power-management techniques available on modern hardware. In the most basic setting, the processor (or device) can reside either in an \emph{active} state in which processing can take place, or in a \emph{sleep} state of negligible energy-consumption in which no processing can take place. Since transitioning the processor from the sleep to the active state requires energy, the scheduler would like to satisfy all the processing requirements, while making use of the sleep state as efficiently as possible. To give some intuition, it is preferable to reside in the sleep state for fewer but longer time-intervals than frequently switching between the two states.

A bit more formally, consider a set of $n$ jobs, each with a release time, a deadline and a processing-time requirement, to be feasibly scheduled on either a single or a multi-processor system that is equipped with a powerdown mechanism. When the (each) processor resides in the active state it consumes energy at a rate of one, whereas it does not consume any energy when in the sleep state. Transitioning from the sleep state to the active state requires a constant amount of energy (the \emph{wake-up cost}), whereas transitioning from the active state to the sleep state is free of charge (this is w.\,l.\,o.\,g., since positive energy cost could be folded onto the wake-up cost). Jobs can be preempted and (in the multi-processor setting) migrated, but every job can be processed on at most one machine at any given time. The objective is to produce a feasible schedule (assuming that such a schedule exists) which consumes the minimum amount of energy.  In Graham's notation, and with $E$ being the appropriate energy function, the problems we study can be denoted as $1|r_j;\overline{d}_j;\text{pmtn}|E$ and $m|r_j;\overline{d}_j;\text{pmtn}|E$ respectively.

The problem was first stated in its \emph{single-processor} version by Irani et al.\,~\cite{iraniSGtalg07} along with an $O(n\log n)$-time $2$-approximation algorithm for it. The algorithm, called \emph{Left-to-Right} greedily keeps the processor at its current state for as long as possible. The problem plays a central role in the area of energy efficient algorithms~\cite{IraniP05}, and the exact complexity of it was only resolved by Baptiste et al.\,~\cite{baptiste2007polynomial} (and a bit earlier for the special case of unit-processing-time jobs~\cite{baptiste2006scheduling}), who gave an exact polynomial-time algorithm for the problem. Their algorithm is based on a dynamic programming approach, which at least for the arbitrary processing-time case is rather involved, and obtains a running time of $O(n^5)$. 

When considering the \emph{multi-processor} setting, it is unclear how to adapt the aforementioned dynamic programming approach to the \emph{multi-processor} setting, while enforcing that a job does not run in parallel to itself. Additionally, several structural properties that had proven useful in the analysis of the single-processor setting do not carry over to the multiprocessor setting. Only very recently, an approximation algorithm for the problem that attains a non-trivial approximation guarantee was presented by Antoniadis et al.~\cite{AGKK20}. Their algorithm is based on carefully rounding a relaxation of a linear programming formulation for the problem and has a approximation ratio of $3$. They also show that their approach gives an LP-based  $2$-approximation algorithm for the single-processor case. We note that whether the problem is NP--hard or not in the multiprocessor-setting remains a major open question.

\subsection{Formal Problem Statement and Preliminaries}

Consider a set of jobs $\{j_1,j_2,\ldots, j_n\}$; job $j_i$ has release time $r_i$, deadline $d_i$ and a processing requirement of $p_i$, where  all these quantities are non-negative integers.
Let $r_{min}$ and $d_{max}$ be the earliest release time and furthest deadline of any job; it is no loss of generality to assume $r_{min} = 0$ and $d_{max}=D$. For $t\in \mathbb{Z}_{\geq 0}$, let $[t,t+1]$ denote the $t^{\rm th}$ {\em time-slot}. Let $I=[t,t'], t, t'\in \mathbb{Z}_{\geq 0}, t<t'$ be an {\em interval}. The length of $I$, denoted by \abs{I} is $t'-t$. Finally, we use $t\in I=[a,b]$ to denote $a\le t\le b$ and call interval $[r_i,d_i]$ the \emph{span} of job $i$.

Two intervals $I_1=[a_1,b_1]$ and $I_2=[a_2,b_2]$ {\em overlap} if there is a $t$ such that $t\in I_1$ and $t\in I_2$. Thus two intervals which are right next to each other would also be considered overlapping. Intervals which do not overlap are considered {\em disjoint}. $I_1$ is {\em contained} in $I_2$, denoted $I_1\subseteq I_2$, if $a_2\le a_1 < b_1\le b_2$ and it is {\em strictly contained} in $I_2$, denoted $I_1\subset I_2$, if $a_2 < a_1 < b_1 < b_2$. 


At any time-slot, a processor can be in one of two states: the \emph{active}, or the \emph{sleep} state. For each time-slot that a processor is in the active state, it requires one unit of power whereas no power is consumed in the sleep state. However, $q$ units of energy (called \emph{wake up} energy) are expended when the processor transitions from the sleep to the active state. In its active state, the processor can either process a job (in which case we refer to it as being \emph{busy}) or just be \emph{idle}. On the other hand the processor cannot perform any processing while in the sleep state. Note that whenever a period of inactivity is at least $q$ time-slots long then it is worthwhile to transition to the sleep state, whereas if it is less than $q$ time-slots long then it is preferable to remain in the active state.


A processor can process at most one job in any time-slot and a job cannot be processed on more than one processor in a time-slot. However, job preemption and migration are allowed, i.e., processing of a job can be stopped at any time and resumed later on the same or on a different processor. A job $j_i$ must be processed for a total of $p_i$ time-slots within the interval $[r_i,d_i]$. Any assignment of jobs to processors and time slots satisfying the above conditions is called a (feasible) \emph{schedule}. We assume that the processor is initially in the sleep state. Therefore, the energy consumed by a schedule is the total length of the intervals during which the processor is active plus $q$ times the number of intervals in which the processor is active. The objective of the problem is to find a schedule which consumes minimum energy.

We will use $P$ to denote the sum of all the processing times, ie. $P=\sum_{i=1}^{n}p_i$. We will use $OPT$ to denote a fixed optimal solution and $Q$ to denote the total wake up cost incurred by $OPT$. Given a solution $S$ (not necessarily feasible), by \emph{maximal} \emph{gap} of $S$, we will refer to intervals $[a,b]$ such that there are no active time slots in $[a,b]$ and processor is active at time slots $a-1$ and $b$. An interval $[a,b]$ of $S$ is \emph{active} if the processor is active in all time slots in $[a,b]$ and inactive in time slots $a-1$ and $b$.

Given a single processor instance, a schedule is called a \emph{skeleton} if for all jobs $j_i$, there is at least one active interval overlapping with the span of $j_i$. In other words, there must be at least one active time slot in $[r_i-1,d_i+1]$. Note that a skeleton need not be a feasible solution. A \emph{minimum cost skeleton} is a skeleton of minimal energy consumption over all skeletons.

In the following we let $(x)^{+}$ to stand for $x$ if $x \geq 0$ and $0$ otherwise. Due to space constraints omitted proofs are deferred to the appendix.


\subsection{Our Contribution}

We study the problem in both the single-processor and the multi-processor setting. Our core technical contribution is that of introducing the concept of minimum cost skeletons. We then employ this concept to design combinatorial and efficient approximation algorithms for single and multi-processor setting. Finally we demonstrate how skeletons can also be useful in strengthening the LP-formulation of the problem with additional constraints by giving the first known $(2+\epsilon)$-approximation algorithm for the multi-processor case. More formally, our contribution is based on the following.

\paragraph*{Single Processor.}
We begin by introducing the notion of skeletons for single-processor instances, and presenting a simple dynamic programming algorithm for computing minimum cost skeletons in Section~\ref{sec:single-skeleton}. Roughly speaking, a skeleton is a (not necessarily feasible) solution which overlaps with the span of every job. Apart from providing a lower bound for the cost of the optimal solution $OPT$, a skeleton has useful structural properties that allow us to convert it into a feasible solution without much overhead. The first result in which we demonstrate this, is the following.

\begin{theorem} \label{theorem:OPT_plus_P}
There exists an $O(n\log n)$-time algorithm that computes a solution of cost at most $OPT+P$. 
\end{theorem}

The algorithm  produces a solution of cost at most $OPT+P$, where $OPT$ is the cost of optimal solution and $P$ is the sum of processing times. Since $P\le OPT$ this implies another $2$-approximation algorithm, thus matching algorithm Left-to-Right of~\cite{iraniSGtalg07} in both running time and approximation guarantee.

We further build upon the ideas of Section~\ref{sec:single-skeleton} in Section~\ref{sec:single-improved} to give a $O(n\log n)$-time algorithm that also builds upon a minimum cost skeleton in order to compute a solution of cost at most $OPT$ $+P/(\alpha-1)+Q(2\alpha+1)$. Here $Q$ is the total wake up cost incurred by the optimal solution and $\alpha$ is any real number greater than 1. We show how this result can be used in order to obtain a $O(n\log^{2}n)$ algorithm that computes a near optimal solution when $P>>Q$ (which in most scenarios is the practically relevant case) or $Q >>P$:
\begin{theorem}
\label{thm:near-optimal}
Let $t=\max \{P/Q,Q/P\}$. Then there exists a $O(n \log^{2}n)$ time algorithm which computes a solution of cost at most $OPT(1+8t^{-1/2})$.
\end{theorem}
Note, that this implies a $(1+\epsilon)$-approximation algorithm when $t \ge 8/\epsilon^2$, whereas Left-to-Right by Irani et al.~\cite{iraniSGtalg07} remains a $2$-approximation algorithm even in that case.

Finally, we give an algorithm that is also an $35/18\approx 1.944$-approximation algorithm in $O(n\log n)$-time, improving upon the approximation-ratio of Left-to-Right algorithm upon the greedy 2-approximation algorithm of Irani et al.~\cite{iraniSGtalg07}:
\begin{theorem}
\label{thm:improving2}
There is a $\frac{35}{18}$ approximation algorithm for the single processor case in $O(n \log n)$ time. 
\end{theorem}

Although (as already mentioned) Baptiste et al.~\cite{baptiste2007polynomial} present an exact algorithm for the problem, it is based on a rather involved dynamic program. In contrast, our algorithms have as their main advantages that they are combinatorial in nature, simple to implement, it has an improved approximation-ratio compared to all other non-exact algorithms for the problem and even obtains a near-optimal solution for the interesting and practically relevant case of $P>>Q$ in near linear time.

\paragraph*{Multiple Processors.} 
Although the notion of a skeleton does not naturally extend to the multi-processor setting, it is possible to define skeletons for that setting so as to capture the same intuition. In Section~\ref{sec:multi-skeletons} we do exactly that before presenting a polynomial-time algorithm for computing minimum-cost multi-processor skeletons. In Section~\ref{sec:multi-converting}, we give an algorithm to convert any skeleton into a feasible solution while increasing the cost by a factor of at most $6$:
\begin{theorem}
    \label{thm:6apx}
    There exists a combinatorial $6$-approximation algorithm for the multi-processor case of the problem.
\end{theorem}
This implies the first combinatorial constant-factor approximation algorithm for the multiprocessor case. The arguments required in its analysis are however much more delicate and involved than the single processor case and heavily build upon the tools developed in~\cite{AGKK20}. 

Finally, in Section~\ref{sect:last_one} we further demonstrate the power of skeletons by using them to develop a $(2+\epsilon)$-approximation algorithm for the multi-processor case. Thus we improve upon the recent breakthrough  $3$-approximation from~\cite{AGKK20}. We note, that obtaining any non-trivial approximation guarantee in the multiprocessor setting has been a long standing open problem (see~\cite{baptiste2007polynomial}).
\begin{theorem}
    \label{thm:lp_2}
    There exists a $(2+\epsilon)$-approximation algorithm for the problem on parallel machines.
\end{theorem}

More specifically, we are able to strengthen the linear program used in~\cite{AGKK20} with additional constraints that are guided by the definition of skeletons. We note that it is unclear whether the rounding scheme used in~\cite{AGKK20} can take advantage of these new constraints. To that end we devise a novel rounding technique, and also show that the considered linear program has an integrality gap of at most $2$. 
In other words, we show the following result.
\begin{theorem}
    \label{thm:pseudo}
There exists a pseudo-polynomial time $2$-approximation algorithm  for deadline scheduling on parallel processors.
\end{theorem}
Theorem~\ref{thm:pseudo} when combined with standard arguments, which were already presented in~\cite{AGKK20}, then implies Theorem~\ref{thm:lp_2}.

\subsection{Further Related Work}


The single-processor setting has also been studied in combination with the other popular power-management mechanism of \emph{speed-scaling} where the processor can additionally vary its speed while in the active state where the power consumption grows convexly in the speed. This allows for increased flexibility, as in some cases it may be beneficial to spend some more energy by increasing the speed in the active state in order to incur larger savings by transitioning the processor to the sleep state for longer periods of inactivity. This technique is commonly referred to as \emph{race to idle} in the literature. The combined problem is known to be \NP-hard~\cite{AlbersA14racetoidle1,KumarS15a} and to admit a fully polynomial time approximation scheme (FPTAS)~\cite{AntoniadisHO15racetoidle2}.

Finally, the problem of minimizing the number of gaps in the schedule, i.\,e., the number of contiguous intervals during which the processor is idle. Note that with respect to exact solutions our problem generalizes that of minimizing the number of gaps. Chrobak et al.\,~\cite{chrobak2017greedy} present a simple $O(n^2\log n)$-time, $2$-approximation algorithm for the problem on single processor with unit-processing times. Demaine et al.\,~\cite{demaine2007scheduling} give an exact algorithm for the problem of minimizing the number of gaps in the multi-processor setting with unit-processing-times.

\section{Computing Minimum Cost Skeletons for a Single Processor}
\label{sec:single-skeleton}

The goal of this section is to prove Theorem~\ref{theorem:OPT_plus_P}. Any solution to the minimum cost skeleton problem can be seen as a (non-overlapping) set of active intervals separated by a set of maximal gaps. In this spirit one interpretation of a skeleton is that of a set of active intervals which overlaps with the span of every job. An equivalent definition for a skeleton is a set of maximal gaps such that the span of no job is properly contained inside any maximal gap: if there is at least one active time-slot in $[r_i-1,d_i+1]$ then the span of job $i$ is not contained in a maximal gap and vice versa. As we will see, this alternative viewpoint will be useful in designing near linear time algorithm for computing minimum cost skeleton.

\begin{definition}
\textbf{Gap Skeleton Problem}: Find a set of maximal gaps $G_1,G_2,\ldots,G_k$ such that for each job $j_i, [r_i-1,d_i+1] \not \subseteq G_p$ for $p=1,2,\ldots,k$ and the quantity $\sum_{i=1}^{k}(|G_i|-q)^{+}$ is maximized.
\end{definition}

Observe that we may  w.\,l.\,o.\,g.\, only consider minimum cost skeletons, the leftmost interval of which begins at $d_{min}$ and the rightmost active interval ends at $r_{max}$ (or else we can transform them to a skeleton satisfying the property without increasing their cost). We call such skeletons $\emph{nice}$. For the purpose of computing a minimum cost skeleton, we shall restrict our attention to nice skeletons only. The maximal gaps of a nice skeleton form a feasible solution to the gap skeleton problem. By construction, the sum of costs of a nice skeleton and the corresponding gap skeleton is exactly equal to $r_{max}-d_{min}+q$. Hence, the problem of finding a minimum cost skeleton is in fact equivalent to finding a maximum cost gap skeleton. We now show how to compute a maximum gap skeleton by using dynamic programming.

Without loss of generality, we may also assume that maximal gaps in any maximum cost gap skeleton start and end at one of the points in $T=\cup_{i=1}^{n}\{r_i,d_i\}$ (otherwise we could increase the length of maximal gaps, without thereby decreasing the cost of the solution). A maximal gap $[x,y]$ is called $\emph{right maximal}$ if there exists a $j_i$ such that $d_i=y$ and $r_i >x$. Without loss of generality, we may assume that there exists an optimal solution to maximum gap skeleton problem in which all the gaps are right maximal (we can always convert any given optimal solution to one containing only right maximal gaps). Let us rename $T=\cup_{i=1}^{n}\{r_i,d_i\}$ to $T=\{t_1,t_2,\ldots,t_{2n}\}$ such that $t_1 \leq t_2 \leq \ldots \leq t_{2n}$. Observe that for any $t \in T$, at most one right maximal gap can have its left endpoint at $t$ and hence there can be a total of at most $2n$ right maximal gaps.

We can list all the right maximal gaps in $O(n \log n)$ time as follows: we sort all the $r_i$'s in $O(n \log n)$ time and then for each $t \in T$, we compute the first $r_i$ to the right of $t$. Then $[t,d_i]$ is the unique right maximal gap starting at $t$. Since for each $t$, the above can be done in $O(\log n)$ time (by using binary search), all the right maximal gaps can be computed in $O(n \log n)$ time.

Let $T_i=\{t_i,t_{i+1},\ldots,t_{2n}\}, i=1,\dots 2n$ and $A[i]$ be the maximum value of the gap skeleton problem when restricted to $T_i$. By the discussion above, we may restrict our attention to only the right maximal gaps. Observe that $A[2n]=0$ and $A[1]$ gives the value of maximum cost gap skeleton. $A$ satisfies the following recurrence: let $g=[t_i,t_j]$ be the right maximal gap starting at $t_i$. In the optimal solution, either a right maximal gap starts at $t_i$ or it doesn't, giving $A[i]=\max \{ (t_j-t_i-q)^{+}+A[j+1],A[i+1]\}$. Using the above recurrence, $A[1],\ldots,A[2n]$ can be computed in $O(n)$ time. Hence, by the equivalence of the two problems, a skeleton with minimum cost can be computed in $O( n \log n)$ time. 
\begin{theorem} \label{theorem:min_cost_skeleton}
A minimum cost skeleton can be computed in $O(n\log n)$ time.
\end{theorem}
Since any feasible solution to the minimum energy scheduling problem is also a skeleton, the following follows:
\begin{observation}
\label{claim:cost_skeleton_OPT}
The minimum cost skeleton has value at most OPT.
\end{observation}
In Lemma \ref{lemma:extending_skeleton} we show how to convert a skeleton into a feasible solution in $O(n \log n)$ time with an additional cost of at most $P$. Along with Theorem \ref{theorem:min_cost_skeleton} and Observation \ref{claim:cost_skeleton_OPT}, this completes the proof of Theorem~\ref{theorem:OPT_plus_P}

\begin{lemma} \label{lemma:extending_skeleton}
Let $S$ be any feasible skeleton and $P_S$ be the maximum total volume of jobs that can be feasibly processed in it. Then we can convert $S$ into a feasible solution $S'$ with an additional cost of $P-P_S$ in $O(n \log n)$ time. 
\end{lemma}

\section{Improved Approximation Algorithms}
\label{sec:single-improved}

We further develop the ideas introduced in the last section to give fast and improved approximation algorithm for the minimum energy scheduling problem. The main insight is to compute a minimum cost skeleton after scaling the wake up cost and then using Lemma \ref{lemma:extending_skeleton} to find a feasible solution. As we will show, this leads to near optimal solutions in case $P>>Q$ or $P<<Q$.

Let $\alpha >1$ be a real number and $S_{\alpha}$ be the minimum cost skeleton obtained by scaling the wake up cost by $\alpha$. $S_{\alpha}$ can be computed in $O(n \log n)$ time by Theorem \ref{theorem:min_cost_skeleton}. Let $F_\alpha$ be the solution obtained by converting $S_{\alpha}$ into a feasible solution using Lemma \ref{lemma:extending_skeleton}. The following theorem bounds the cost of $F_{\alpha}$ in terms of $P,Q$ and $\alpha$.

\begin{theorem}\label{theorem:skeleton_scaled}
The cost of $F_{\alpha}$ is at most $OPT+ 2(\alpha+1)Q+P/(\alpha-1)$.
\end{theorem}
\begin{proof}
Let $OPT$ be a fixed optimal solution to the original instance. We abuse notation and also use $OPT$ to denote the cost of the solution. To bound the cost of $F_\alpha$, we will first convert $OPT$ into $S_{\alpha}$ in a series of steps, while carefully accounting for changes in the cost, and the total volume of jobs that can be processed. Since $F_\alpha$ can be obtained from $S_\alpha$ by using Lemma~\ref{lemma:extending_skeleton} with a further increase in cost equal to the missing volume, this will allow us to obtain the desired bound. Let $g=[a,b]$ be any maximal gap in $OPT$. We first show that at most two active intervals of $S_{\alpha}$ overlap with $g$. The proof of this claim follows a similar proof found in Irani et al.~\cite{iraniSGtalg07}.

\begin{claim} \label{clm:at_most_two}
At most two active intervals of $S_{\alpha}$ overlap with any maximal gap $g$ of $OPT$.
\end{claim}

\begin{proof}
The proof of this claim is inspired by a similar proof found in Irani et al.~\cite{iraniSGtalg07}. For the sake of contradiction, assume that there are three active intervals $I_1,I_2,I_3$ in $S_\alpha$ which overlap with $g$. Let $I_1=[a_1,b_1],I_2=[a_2,b_2],I_3=[a_3,b_3]$ and $a_1 <b_1 < a_2 < b_2 < a_3 < b_3$. We show that removing $I_2$ from $S_\alpha$ doesn't destroy its feasibility, thus implying that $S_\alpha$ is not the minimum skeleton, a contradiction.

First observe that $I_2$ is completely contained inside $g$. If not, $I_2$ must contain one of the endpoints of $g$ (since it overlaps with $g$), implying that there is no interval either to the left or the right of $I_2$ overlapping with $g$, a contradiction. To complete the proof, we now show that there is no job $k$ whose span overlaps only with $I_2$. This would imply that $S_\alpha$ continues to be a valid skeleton even after removal of $I_2$, contradicting its optimality.

For the sake of contradiction, assume that there exists a job $j_k$ such that its span overlaps with $I_3$, but doesn't overlap with either $I_1$ or $I_2$.
\begin{itemize}
    \item span of $j_k$ is strictly contained completely inside $I_2$, ie. $[r_k,d_k] \subset I_2$. Since $I_2 \subseteq g$, this implies that $[r_k,d_k] \subset g$, a contradiction (as $j_k$ is not scheduled in $OPT$ at all).
    \item span of $j_k$ overlaps with the left end point of $I_2$ ie. $a_2 \in [r_k,d_k]$. In this case, $ b_1 < r_k < a_2$ and $a_2 < d_k < a_3$. Once again, this implies that $[r_k,d_k]\subset g$, a contradiction.
    \item span of $j_k$ overlaps with the right end point of $I_2$ ie. $b_2 \in [r_k,d_k]$. In this case, $ b_1 < r_k < b_2$ and $b_2 < d_k < a_3$. Once again, this implies that $[r_k,d_k]\subset g$, a contradiction.
\end{itemize} 
\end{proof}

Let $g$ be a maximal gap of $OPT$ and $I_1,I_2$ be the two intervals of $S_\alpha$ that overlap with $g$ (it is of course possible that either one or both of these intervals are empty). Recall that $|I_1 \cap g|$ and $|I_2 \cap g|$ denote the length of overlap between $g$ and $I_1,I_2$ respectively.
\begin{claim}
$|I_1\cap g| \leq \alpha q$ and $|I_2\cap g| \leq \alpha q$. 
\end{claim}
\begin{proof}
Let $I_1 \cap g=[a,b]$. For the sake of contradiction, suppose that $|I_1 \cap g| > \alpha q$. Observe that span of no job is strictly contained inside $I_1 \cap g$, ie. for no $j_k$, $[r_k,d_k] \subset I_1 \cap g$ (otherwise job $j_k$ would not be scheduled in $OPT$ at all). This implies that $S_\alpha$ remains a valid skeleton even if we make all the time slots in $[a,b]$ inactive. This operation decreases the total length of active intervals in $S_{\alpha}$ by $|a-b|>\alpha q$ while introducing at most one new active interval. This implies that $S_\alpha$ is not the optimal skeleton with wake up cost equal to $\alpha q$, a contradiction (as making all time slots in $[a,b]$ will give a lower cost skeleton). Hence, $|I_1\cap g| \leq \alpha q$. A similar argument shows that $|I_2 \cap g|\leq \alpha q$.   
\end{proof}

We first transform $OPT$ so that its active intervals contain the active intervals of $S_{\alpha}$. More formally:

\begin{claim} \label{claim:mod_OPT_1}
We can modify $OPT$ while increasing the cost by at most $2(\alpha+1)Q$ so that for any active interval $I_\alpha \in S_\alpha$, there exists an active interval $I_o \in OPT$ such that $I_\alpha \subseteq I_o$.
\end{claim}
\begin{proof}
We transform every maximal gap $g$ of $OPT$ as follows. Let $I_1,I_2$ be two active intervals of $S_\alpha$ that overlap with $g$. We add (active) intervals $I_1 \cap g$ and $I_2 \cap g$ to $OPT$. The additional cost incurred is $2q+|I_1 \cap g|+|I_2 \cap g| \leq 2q+\alpha q +\alpha q=2(\alpha +1)q$. The condition of the claim clearly holds after we perform the above transformation for every maximal gap of $OPT$. The total cost incurred over all maximal gaps is $2(\alpha +1)Q$ and the claim follows.
\end{proof}

Let $OPT_1$ be the solution obtained after the modification in Claim \ref{claim:mod_OPT_1}. Let $OPT_1 \setminus S_\alpha$ denote the set of active intervals formed by taking the difference of active time slots in $OPT_1$ and $S_\alpha$. Let $G_\alpha$ be the set of all maximal gaps of $S_\alpha$. By Claim \ref{claim:mod_OPT_1}, for every active interval $I \in OPT_1 \setminus S_\alpha$, there exists a maximal gap $g_\alpha \in G_\alpha$ such that $I \subseteq g_\alpha$. We convert $OPT_1$ into $S_\alpha$ by removing all the active intervals in $OPT_1 \setminus S_\alpha$. To bound the cost of $OPT_1$, we need to consider the following possibilities:
\begin{itemize}
    \item $I \subset g_\alpha$: In this case, removing $I$ doesn't create any new active intervals in $OPT_1$. Hence, decrease in the cost of $OPT_1$ is equal to the length of $I$. Also, the reduction in the total volume of jobs that can be scheduled in $OPT_1$ is at most $|I|$.  
    
    \item $I = g_\alpha$: In this case, removing $I$ creates a new interval in $OPT_1$. Hence, total cost of $OPT_1$ decreases by $|I|-q$. Also, reduction in the total volume of jobs that can be scheduled in $OPT_1$ is at most $|I|$.
    
\end{itemize}
Let $A$ the set of intervals such that $I\in OPT_1\setminus S_\alpha = g_\alpha$. We need the following bound on the cardinality of $A$ to finish the proof.
\begin{claim}
$q|A| \leq P/(\alpha-1)$.
\end{claim}

\begin{proof}
Recall that we compute $S_\alpha$ by scaling the wake up cost by $\alpha$, hence any maximal gap of $S_\alpha$ has length at least $\alpha q$. Since for any interval $I \in A$, there exists a maximal gap $g_\alpha \in S_\alpha$ such that $I=g_\alpha$, we have $|I|\geq \alpha q$ for any $I \in A$. We now show that there can't be more than $q$ idle time slots in $I$ in any feasible schedule of $OPT$.

Fix a feasible schedule of jobs in $OPT$. Since $I=g_\alpha$ and $S_\alpha$ is a feasible skeleton, there is no $k$ such that span of $j_k$ is strictly contained in $I$. Let $J_1$ be the set of jobs which overlap with the left end point of $I$ and $J_2$ be the set of jobs which overlap only with the right end point of $I$. Note it is possible that some job is included in both $J_1$ and $J_2$. By shifting the jobs in $J_1$ to as far left as possible and those in $J_2$ to as far right as possible, we may assume that all the idle slots in $I$ appear contiguously. If the length of this idle interval is more than $q$, then its removal decreases the cost of the solution without affecting the feasibility, thus contradicting the fact that $OPT$ is a minimum cost solution. Hence, there cannot be more than $q$ idle slots in any $I \in A$ in any feasible schedule of $OPT$. This implies that at least $|I|-q$ volume of jobs is processed in $I$ in any feasible schedule of $OPT$. Hence, $P \geq \sum_{I \in A}(|I|-q) \geq (\alpha-1) q |A|$ and the statement of the claim follows.  \end{proof}

We are now ready to bound the cost of the final solution. By the case analysis above, the cost of $S_\alpha$ is at most $OPT_1 -  \sum_{I \in OPT_1 \setminus S_\alpha}|I|+q|A| \leq OPT +2(\alpha+1)Q - \sum_{I \in OPT_1 \setminus S_\alpha}|I|+q|A|$. The total volume of jobs that can be processed in $S_\alpha$ is at least $P- \sum_{I \in OPT_1 \setminus S_\alpha}|I|$. Hence, using Lemma \ref{lemma:extending_skeleton} to convert $S_\alpha$ into a feasible solution gives that the total cost of $S_\alpha$ is at most $OPT +2(\alpha+1)Q -  \sum_{I \in OPT_1 \setminus S_\alpha}|I|+q|A| + \sum_{I \in OPT_1 \setminus S_\alpha}|I|=OPT + 2(\alpha +1) Q+ q|A| \leq OPT + 2(\alpha +1) Q+ P/ (\alpha-1)$. This completes the proof of the theorem.
\end{proof}

We are now ready to prove the two mains theorems of this section. We start with Theorem~\ref{thm:near-optimal}.

\begin{proof}[Proof of Theorem~\ref{thm:near-optimal}]
We construct a series of solutions, $F_0,F_1,F_2,F_4,\ldots,F_{2^{\log \lceil n \rceil}}$ as follows: $F_0$ is the solution given by Theorem \ref{theorem:OPT_plus_P} of cost at most $\tt{\tt{OPT}+P}$. For each $i\geq 1$, we construct $F_i$ by setting $\alpha=\sqrt{P/iq}$ in Theorem \ref{theorem:skeleton_scaled}. Our final solution $F$ is obtained by taking the minimum cost solution among all the $F_i's$. Since, each of $F_i's$ can be constructed in $O(n \log n)$ time, $F$ can be constructed in $O(n \log^{2}n)$ time. We now show that $F$ has the desired approximation guarantee. First note that $OPT \geq P+Q$ and $t\geq 1$. If $P\leq Q$, then $F_0$ has cost at most $OPT+P = OPT(1+\frac{P}{OPT})\leq OPT(1+\frac{P}{P+Q}) \leq OPT(1+t^{-1}) \leq OPT(1+8t^{-1/2})$. Now consider the case when $P > Q$. One can verify that for $ \alpha \geq 1,f(\alpha)= Q (2\alpha+2) +P/(\alpha -1)$ has a unique minimum at $\alpha^{*}=\sqrt{P/2Q}+1$. We must have used exactly one $\alpha \in [\alpha^{*},\sqrt{2}\alpha^{*}]$ to construct one of the $F_i$'s. Since, $f(\alpha)$ has a unique minimum in $\alpha \geq 1$, it follows that it is increasing in $[\alpha^{*},\sqrt{2}\alpha^{*}]$. Hence, there exists a $F_i$ with cost no more than guaranteed by setting  $\alpha=2 \alpha^{*}$ in Theorem \ref{theorem:skeleton_scaled}. A straightforward calculation shows that $f(2\alpha^{*}) \leq 8 \sqrt{PQ}$. Hence, $F$ has cost at most $OPT+8\sqrt{PQ} \leq OPT(1+\frac{8\sqrt{PQ}}{P+Q}) \leq OPT\left(1+ \frac{8\sqrt{t}}{t+1}\right)\leq OPT(1+8t^{-1/2})$.
\end{proof}


Finally, we give a $O(n \log n)$ algorithm that has a performance guarantee better than $2$.

\begin{proof}[Proof of Theorem~\ref{thm:improving2}]
We construct two solutions: first one of cost at most $OPT +P$ (by using Theorem \ref{theorem:OPT_plus_P}) and second one of cost at most $OPT+8Q+P/2$ (by using Theorem \ref{theorem:skeleton_scaled}). In case $P \leq 17 Q$, the first solution has cost at most $OPT+P \leq OPT + 17(P+Q)/18 \leq OPT \cdot 35/18$. In case $P > 17 Q$, the second solution has cost at most $OPT+8Q+P/2 \leq  OPT+8(1/18)\cdot OPT + 1/2 \cdot OPT = (1+8/18 + 1/2) \cdot OPT = OPT \cdot 35/18$. \end{proof}
\section{Skeletons for Parallel Processors}
\label{sec:multi-skeletons}

In this section, we extend the idea of skeletons from the single processor setting to the multi-processor one. We design an efficient and combinatorial algorithm for finding the minimum cost skeleton. In the next section we then show how this skeleton can be used to design a combinatorial approximation algorithm for the multi-processor setting. Let $T$ be as defined in the last section, i.\,e.\, $T=\cup_{i=1}^{n} \{r_i,d_i\}$. For any $t_i,t_j \in T $ such that $t_i < t_j$, let $l(t_i,t_j)$ be the maximum number of processors that can be blacked out in $[t_i,t_j]$. More formally, $l(t_i,t_j)$ is the maximum number of processors so that there exists a feasible schedule using at most on $m-l(t_i,t_j)$ many processors at any timeslot $t\in [t_i,t_j)$.  Equivalently,  at some time $t \in [t_i,t_j]$, at least $m-l(t_i,t_j)$ processors must be active in any feasible solution. We are now ready to define skeleton for the multi-processor case. 
\begin{definition}
A set of active intervals (not necessarily feasible) is called a $\textbf{skeleton}$ if for any time interval $[t_i,t_j]$, there exists a $t \in [t_i,t_j)$ such that at least $m-l(t_i,t_j)$ processors are active at timeslot $t$.    
\end{definition}

By the definition of $l(t_i,t_j)$ and the definition of multi-processor skeletons, it directly follows that every feasible solution is also a skeleton. Hence, the cost of the optimal skeleton is a lower bound on the cost of an optimal solution.

We note that all $l(t_i,t_j)$'s can be computed in polynomial time: there are $O(n^{2})$ possible pairs $(t_i,t_j)$ and for each pair, the value $l(t_i,t_j)$ can be computed in $O(F \log m)$ time by using binary search ($F$ here denotes the time needed to check feasibility of an instance, which can also be done in polynomial time as we will see in the next section). Therefore, the total time required to compute all $l(t_i,t_j)'s$ is $O(Fn^{2}\log m)$. 

\subsection{Computing Minimum Cost Skeleton for Parallel Processors}

We show how an optimal multi-processor skeleton can be computed by combining up to $m$ many distinct single-processor skeletons. To that end, let $\mathcal{I}_k$ be the set of all tuples $(t_i,t_j)$ such that $m-l(t_i,t_j)\geq k$. Each of $\{I_k\}_{k=1}^{m}$ can be thought of as defining an instance of the minimum skeleton problem for a single processor as follows: for each $I=[a,b] \in \mathcal{I}_k$, we have a job with release time $a$, deadline $b$ and a unit processing requirement. We can compute the minimum skeleton $S_k$ for each $k=1,2,\ldots,m$ using Theorem \ref{theorem:min_cost_skeleton}. It remains to show that $\{S_k\}_{k=1}^{m}$ is indeed the desired optimal skeleton.

\begin{lemma} \label{lemma:skeleton_mult_proc}
\label{cl:opt-skeleton-m}
$\{S_k\}_{k=1}^{m}$ is an optimal skeleton for the multi-processor case.
\end{lemma}
\begin{proof} Let $O$ be the optimal multi-processor skeleton for an arbitrary given instance. Without loss of generality, we may assume that $O$ has a laminar structure, i.\,e.\, each active interval on processor $k+1$ is a subset of some active interval on processor $k$. Let $O_k$ be the set of active intervals on processor $k$ in the optimal solution and $l_i$ be the total length of active intervals on processor $i$.
We consider $O = \{O_k\}_{k=1}^m$ such that it is lexicographically maximal with respect to the $m$-tuple $(l_1,l_2,\ldots,l_m)$, and it differs from $\{S_k\}_{k=1}^{m}$ in least number of processors among those lexicograhically maximal ones. If $O_k=S_k$ for $1 \leq k \leq m$, then the lemma follows. For the sake of contradiction, let us assume that $O_k \neq S_k$, for some $k$.

\begin{claim}
$O_r$ is a feasible skeleton for the single processor instance $\mathcal{I}_r, 1 \leq r \leq m$.
\end{claim}

\begin{proof}
Suppose the statement of the claim doesn't hold. Then there exists an interval $[t_i,t_j]$ such that $m-l(t_i,t_j)\geq r$ but there is no active interval overlapping $[t_i,t_j]$ in $O_r$. Since the optimal solution $O$ is laminar, there is no active interval overlapping $[t_i,t_j]$ for any $O_l, r \leq l \leq m$. This implies that the number of active interval at any time in $[t_i,t_j]$ in $O$ is at most $r-1 < m-l(t_i,t_j)$. This contradicts the feasibility of the optimal solution and the claim follows.
\end{proof}

\begin{claim}
The solution obtained by replacing the intervals on processor $k$ in the optimal solution, i.\,e.\, $O_k$ by $S_k$ is a feasible multi-processor skeleton.
\end{claim}

\begin{proof}
Suppose the statement of the claim doesn't hold. Then there exists an interval $[t_i,t_j]$ such that some active interval in $O_k$ overlaps with $[t_i,t_j]$ but no active interval in $S_k$ overlaps with $[t_i,t_j]$. Since the optimal solution $O$ is laminar, this implies that some active interval in $O_l$ overlaps with $[t_i, t_j]$ for any $1 \leq l \leq k$. Hence, $m-l(t_i,t_j) \geq k$, which implies that some active interval in $S_k$ must overlap with $[t_i,t_j]$. This contradicts our assumption and the claim follows.
\end{proof}

Since $O_k$ is a feasible solution to $\mathcal{I}_k$ and $S_k$ is an optimal solution for $\mathcal{I}_k$, cost of $S_k$ is at most the cost of $O_k$. Hence, replacing $O_k$ by $S_k$ gives a feasible skeleton without increasing the cost. The new solution as constructed above is laminar as well, otherwise we could move active time slots from a higher numbered processor to a lower numbered processor, contradicting our assumption that the optimal solution is the largest in lexicographical ordering $(l_1,l_2,\ldots,l_m)$. Thus we have obtained a different optimal solution which matches $\{S_i\}_{i=1}^{m}$ on more processors. This contradicts our choice of the optimal skeleton and the lemma follows.
\end{proof}

\section {Converting a Minimum Cost Skeleton into a Feasible Solution}
\label{sec:multi-converting}

As argued in Lemma~\ref{cl:opt-skeleton-m}, the cost of the obtained optimal skeleton is at most the cost of an optimal solution. However, the optimal skeleton may not be a feasible solution. In this section we show how to overcome this by transforming the optimal skeleton $\{S_k\}_{k=1}^m$ into a feasible solution while increasing its energy cost by at most a factor $6$.
This transformation consists of two phases: the \emph{extension phase} and the \emph{tripling phase}. In the following subsections we describe each one of them in more detail.

\subsection{Extension Phase}

The extension phase of the transformation is inspired by a similar transformation performed in~\cite{AGKK20}. For the sake of completeness we give a brief and high-level description of the required terminology and results and refer the interested reader to~\cite{AGKK20} for the details. 

We begin by introducing the notions of \emph{forced volume}  and of \emph{deficiency}:

\begin{definition}[\cite{AGKK20}]
The \emph{forced volume} of a job $j_i$ with respect to an interval $[a,b]$, is defined as $\texttt{fv}(j_i,[a,b]): = \max\{0,p_i - (|[r_i,d_i]\setminus [a,b]|)$\}.  Let $\mathcal{D}$ be a set of disjoint intervals. The \emph{forced volume} of job $j_i$ with respect to $\mathcal{D}$ is defined as $\texttt{fv}(j_i,\mathcal{D}):= \max\{0,p_i- (|[r_i,d_i]| -\sum_{D\in \mathcal{D}}|D\cap[r_i,d_i]|)\}$.
\end{definition}

Intuitively, $\texttt{fv}(j_i,[a,b])$ is the minimum volume of $j_i$ that must be processed during $[a,b]$ in any feasible schedule, and $
\texttt{fv}(j_i,\mathcal{D})$ is the amount of volume that must be processed within the intervals of $\mathcal{D}$ in any feasible schedule.

\begin{definition}[\cite{AGKK20}]
    Let $\mathcal{D}$ be a set of disjoint intervals, and $\mathcal{I}=\{I_1,I_2,\dots I_k\}$ be a set of not necessarily disjoint intervals with the property, that for any time-point $t$, $m_t := |\{i\in\mathcal{I}: I_i\cap t\neq \emptyset\}|\le m$ holds. Furthermore let $\mathcal{J}$ be a set of jobs.  The  \emph{deficiency} of $\mathcal{D}$ \emph{with respect to} $\mathcal{I}$ and $\mathcal{J}$, denoted by \texttt{def}$(\mathcal{D},\mathcal{I},\mathcal{J})$, is the non-negative difference between the sum of the forced volume of all jobs of $\mathcal{J}$ with respect to $\mathcal{D}$ and the total volume that can be processed in $\mathcal{D}$ within $\mathcal{I}$. Thus
        \begin{align*}
            \texttt{def}(\mathcal{J},D,\mathcal{I}) = \max \left(0,\sum_{j\in\mathcal{J}}\texttt{fv}(j,\mathcal{D}) - \sum_{t:[t,t+1]\subseteq \mathcal{D}} m_t \right).
        \end{align*}
\end{definition}

In~\cite{AGKK20}, a decision problem called \texttt{deadline-scheduling-on-intervals} was introduced. More formally, problem \texttt{deadline-scheduling-on-intervals} takes as input $k$ (not necessarily disjoint) supply-intervals $\mathcal{I} = \{I_1,I_2,\dots I_k\}$ and the set $\mathcal{J}$ of  jobs (each with a release-time, a deadline and processing volume), and asks whether the jobs of $\mathcal{J}$ can be feasibly scheduled on $\mathcal{I}$. In~\cite{AGKK20} a polynomial-time algorithm \emph{DSI-ALG} was presented that decides \texttt{deadline-scheduling-on-intervals}. Furthermore, in case the input instance is infeasible, DSI-ALG returns a \emph{minimal} set of intervals $\mathcal{D} = \{D_1,D_2,\dots D_\ell\}$ of maximum deficiency with respect to $\mathcal{I}$. The following theorem follows from Section~4 in~\cite{AGKK20} (more specifically the first statement appears in~\cite{AGKK20} as Theorem~$4.1$, the last one as Claim~$4.2$, and the polynomial-time algorithm is described and analyzed throughout Section~$4$):

\begin{theorem}[\cite{AGKK20}]
    \label{thm:feasibility}
    An instance of \texttt{deadline-scheduling-on-intervals} is feasible iff no set of disjoint intervals has positive deficiency. Additionally, there is a polynomial-time algorithm that decides if a given instance to \texttt{deadline-}\texttt{scheduling-}\texttt{on-}\texttt{intervals} is feasible and if it is not, then a minimal set of disjoint intervals of maximum deficiency with respect to the instance is returned. Furthermore, increasing the volume of supply intervals at any time point in the minimal set of maximum deficiency by one unit decreases the maximum deficiency by one unit.
\end{theorem}

Finally, \cite{AGKK20} presents a polynomial time algorithm $\emph{EXT-ALG}$ which extends a supply interval $I'\in \mathcal{I}$ (a property that we will use later, is that $I'$ is chosen so that it overlaps some $D'\in\mathcal{D}$ without containing $D'$, i.e, $D'\not\subseteq I'$ but $D'\cap I'\neq \emptyset$) by one time slot so as to decrease the total deficiency of the set of intervals of maximum deficiency $\mathcal{D}$ by one. This is repeated until the resulting set of supply intervals becomes feasible. Since the maximal deficiency at the beginning was at most the total processing time $P$ of all jobs, the total increase in energy consumption by extending the supply intervals could also only have been at most $P$.

The extension phase consists of repeatedly extending the intervals of $\{S_k\}_{k=1}^m$ via algorithm EXT-ALG until this is not possible anymore. Assume that at this point $\{S_k\}_{k=1}^m$ has been extended to an interval set $\mathcal{I}$. The extension phase thus terminates either because $\mathcal{I}$ is a feasible instance for $\mathcal{J}$ (and by Theorem~\ref{thm:feasibility} no set of disjoint intervals has positive deficiency with respect to $\mathcal{I}$ and $\mathcal{J}$) or because the minimal set $\mathcal{D}$ of maximum deficiency returned by DSI-ALG does not contain any interval $D'$ such that $I'\cup D'\neq\emptyset$ and $I'\not\supseteq D'$ holds for some $I'\in\mathcal{I}$. In the later case $\mathcal{I}$ is still not feasible, and a further transformation (described in the next subsection) is required. The extension phase as stated now is pseudo-polynomial but can be carried out in polynomial time by using standard techniques (see \cite{AGKK20} for more details). From the argument from~\cite{AGKK20} as well as the discussion above, the following lemma follows:

\begin{lemma}
     \label{lem:extension-phase} \label{lemma:multi_proc_phase_1}
     The energy-cost of the schedule $\mathcal{I}$ differs from that of $\{S_k\}_{k=1}^m$ by at most an additive factor of $P$. Furthermore $\mathcal{I}$ is either feasible, or contains no interval $I'\in\mathcal{I}$ that overlaps but does not contain an interval from the minimal set of intervals of maximum deficiency $\mathcal{D}$.
\end{lemma}

\subsection{Tripling Phase}

In case the extension phase terminated with an infeasible solution, then, by Lemma~\ref{lem:extension-phase}, there is no interval $I'\in\mathcal{I}$ such that $I'$ overlaps some interval $D'\in\mathcal{D}$ without containing it. In that case, we need to perform the tripling phase, in which we carefully power on further machines at specific times so as to make the instance feasible. Let $m_t$ be the number of machines active at time $t$ in $\mathcal{I}$. We create a new solution $\mathcal{I'}$ by setting $m_{t}'=\min \{3m_t,m\}$. In Lemma \ref{lemma: multi_proc_feasibility}, we show that $\mathcal{I'}$ is a feasible solution to the original instance. By construction, the total cost of intervals in $\mathcal{I'}$ is at most thrice that of $\mathcal{I}$. From the above discussion and Lemma~\ref{lemma:multi_proc_phase_1}, it follows that the algorithm consisting of the tripling and the extension phase is a $6$-approximation algorithm. In other words Theorem~\ref{thm:6apx} follows.

\begin{lemma} \label{lemma: multi_proc_feasibility}
$\mathcal{I'}$ is a feasible solution to the original instance.
\end{lemma}
\begin{proof}
Suppose $\mathcal{I}'$ (and hence $\mathcal{I}$) is infeasible. Let $\mathcal{D},\mathcal{D}'$ be the set of minimal intervals of maximum deficiency guaranteed by Theorem \ref{thm:feasibility}. Note that as we extend the intervals using $EXT-ALG$ in the extension phase, minimal set of maximum deficiency can only shrink. Hence, every interval in $\mathcal{D'}$ is a subset of some interval in $\mathcal{D}$. Let $\mathcal{D}=\{D_1,D_2,\ldots,D_k\}$ and $\mathcal{D}'=\cup_{i=1}^{k}\mathcal{D}'_i$ where $\mathcal{D}'_i$ is the set of intervals of $\mathcal{D}'$ contained in $D_i$. Let $k_i$ be the number of active intervals in $\mathcal{I}$ overlapping with $D_i$ and $l_i$ be the number of jobs $m$ such that $[r_m,d_m] \subsetneq D_i$ and $fv(j_m,\mathcal{D})>0$. The following claim gives a bound on $l_i$ in terms of $k_i$.

\begin{claim} \label{claim:number_jobs_overlap}
$l_i \leq 2 k_i$.
\end{claim}

\begin{proof}
Let $D_i=[a,b]$ and $l_{i}^{a}$ be the number of jobs $j_m$ such that $[r_m,d_m] \subsetneq D_i$, $fv(j_m,\mathcal{D})>0$ and $a \in [r_m,d_m]$. Let $l_{i}^{b}$ be defined analogously. Then consider the set $\mathcal{D} \cup [a-1,a]$. By definition of forced volume, $fv(J,\mathcal{D}\cup [a-1,a]) \geq fv(J, \mathcal{D}) + l_{i}^{a}$. Also, $\sum_{t:[t,t+1] \subseteq \mathcal{D}U[a-1,a]}m_t=\sum_{t:[t,t+1] \subseteq \mathcal{D}}m_t + k_i$. Hence, $def(J,\mathcal{D}\cup [a-1,a],\mathcal{I}) \geq def(J,\mathcal{D},\mathcal{I})+(l_{i}^{a}-k_i)$. Since, $\mathcal{D}$ is a set of intervals with maximum deficiency for $\mathcal{I}$, it must be true that $l_{i}^{a} \leq k_i$. A similar argument shows that $l_{i}^{b} \leq k_i$. Hence, $l_i \leq l_{i}^{a} +l_{i}^{b}\leq 2k_i$.  
\end{proof}

If $m'_t=m$ for all $[t,t+1] \subseteq D' \in \mathcal{D'}$, then the original instance must be infeasible (by Theorem $\ref{thm:feasibility}$). Hence $m'_t=3m_t< m$ for some $[t,t+1] \subseteq D' \in \mathcal{D}'_i$. By Lemma \ref{lemma:multi_proc_phase_1}, we have that $m'_t=3m_t=3k_i$ for all $[t,t+1] \subseteq D' \in \mathcal{D}'_i$. We finish the proof of the lemma by showing that deficiency of $\mathcal{D}\setminus \mathcal{D}'_i$ is at least the deficiency of $\mathcal{D'}$, thus contradicting the fact that $\mathcal{D'}$ is the minimal set of maximum deficiency returned by Theorem \ref{thm:feasibility}.

\begin{claim}  \label{clm:tripling_final}
$def(J,\mathcal{D'}\setminus \mathcal{D}'_i,\mathcal{I}') \geq def(J,\mathcal{D'},\mathcal{I}')$.
\end{claim}
\begin{proof}
For notational convenience, let $B=\sum_{t:[t,t+1] \subseteq \mathcal{D}'_i}m_t$. By the discussion above, we have $m'_t=3 k_i$ for all $[t,t+1] \subseteq D' \in \mathcal{D}'_i$. Let $J_i$ be the set of jobs such that for each $j_t \in J_i$, $[r_t,d_t] \subseteq D_i$. Since $\mathcal{I}$ is a multiprocessor skeleton, it must be true that $k_i \geq m-l(D_i)$. Hence it follows that all jobs in $J_i$ can be feasibly scheduled in $\mathcal{I}$. By Theorem \ref{thm:feasibility}, we have $fv(J_i,\mathcal{D}'_i) \leq B$.

Let $J_{\mathcal{I}},J_{\mathcal{I'}}$ be the set of jobs which are not completely contained in $D_i$ and have a strictly positive forced volume with respect to $\mathcal{I},\mathcal{I'}$ respectively. Since, every interval in $\mathcal{D}'$ is a subset of some interval in $\mathcal{D}$, forced volume of any job with respect to $\mathcal{I'}$ can not be more than its forced volume with respect to $\mathcal{I}$. Hence, $J_{\mathcal{I}'} \subseteq J_{\mathcal{I}}$. Claim \ref{claim:number_jobs_overlap} gives $|J_{\mathcal{I'}}| \leq |J_{\mathcal{I}'}|\leq 2k_i$. Hence jobs in $j_{\mathcal{I}'}$ contribute a maximum forced volume of $\sum_{t:[t,t+1] \subseteq \mathcal{D}'_{i}} 2k_i=\sum_{t:[t,t+1] \subseteq \mathcal{D}'_{i}} 2 m_t=2B$. 

Observe that removing $\mathcal{D}'_{i}$ doesn't introduce any new jobs with a strictly positive forced volume. Hence, by the discussion above, the total reduction in forced volume due to removal of $\mathcal{D'}_{i}$ is at most $3B$. Removing $\mathcal{D'}_{i}$ from $\mathcal{D}$ reduces the total volume available for processing by $3B$. Hence removing $\mathcal{D'}_{i}$ from $\mathcal{D}$ doesn't result in decrease of the forced volume and the statement of the claim follows.
\end{proof}

\end{proof}

\section{A 2-Approximation Algorithm for Multiple Processors}
\label{sect:last_one}

In this section we prove Theorem~\ref{thm:lp_2} thus improving upon the recent $3$-approximation algorithm of ~\cite{AGKK20}. To achieve this, we introduce additional constraints to the linear programming (LP) relaxation of \cite{AGKK20} and devise a new rounding scheme to harness the power of new constraints. In the remainder of the section, we prove that our rounding technique gives a pseudo-polynomial time $2$-approximation algorithm -- thus also showing an upper bound of $2$ on the integrality-gap of the LP. By using standard arguments, this directly implies $(2+\epsilon)$-approximation algorithm in polynomial time (see \cite{AGKK20} for more details). 

We now give a brief description of the LP relaxation of ~\cite{AGKK20}. For every possible interval $I\subseteq [0,D]$, there is an associated  variable $x_I,\  0\le x_I\le m$ which indicates the number of times $I$ is picked in the solution. The objective is to minimize the total energy consumption, ie. $\sum_I x_I(\abs{I}+q)$. $m_t$ denotes the number of processors that are active during time slot $t$ (or equivalently total capacity of active intervals in time slot $t$) and $f(i,t)$ denotes the volume of job $i$ that is processed in time slot $t$. The constraints of the LP are self explanatory; the interested reader is referred to \cite{AGKK20} for the details.

We now describe the additional constraints (in bold). Recall that for any interval $[a,b]$, $l(a,b) \in \mathbb{Z}_{\ge 0}$ is the maximum number of machines that can remain inactive throughout interval $[a,b]$ without affecting the feasibility of the instance (see Section \ref{sec:multi-skeletons}). This implies that at least $m-l(a,b)$ active intervals are overlapping with $[a,b]$ in any feasible schedule. We add a constraint capturing this fact for every $[a,b] \subseteq [0,D]$. 

\begin{lp}{minimize}{\sum_I x_I(\abs{I}+q)}
\cnstr{m_t}{=}{\sum_{I:[t,t+1]\in I} x_I}{0 \le t < D}
\cnstr{m_t}{\ge}{\sum_{i: r_i\le t\le d_i-1} f(i,t)}{0 \le t < D}
\cnstr{p_i}{=}{\sum_{t=r_i}^{d_i-1} f(i,t)}{0 \le i \le n}

\cnstr{\bm{\displaystyle\sum_{I:[a,b]\cap I\neq\varnothing} x_I}}{\bm{\ge}}{\bm{m-l(a,b)}}{\bm{0 \le a < b\le D}}

\cnstr{f(i,t)}{\in}{[0,1]}{ \forall i,t}
\cnstr{x_I,m_t}{\in}{[0,m]}{\forall t, I\subseteq[0,D]}
\end{lp}
 
Suppose the optimum fractional solution to the linear program has value $f$. For reasons that will become apparent later, we would like to use an optimum solution maximizing the value $\sum_{t} \min(m_t,1)$. In order to compute such a solution we solve a second linear program that is based on the previous one, as follows: we introduce a new set of variables $y_t$ and additional constraints $y_t \leq m_t$ and $0 \leq y_t \leq 1$ for each time slot $t$. By adding constraint $\sum_I x_I(\abs{I}+q)=f$ we enforce that the resulting solution has energy cost equal to $f$ and is therefore also optimal. Finally we set the objective function to maximize $\sum_{t}y_t$. 

Let $F=\{I:x_I>0\}$ be the optimal fractional solution after solving the second LP. Let $\epsilon = \gcd_{i \in F}(x_I)$. We create $x_I / \epsilon$ copies of each $I \in F$ to assume that all intervals in $F$ have the same $x_I$ value (note that $F$ is a multiset). If there exist $[a,b],[c,d], a<c<b<d$ with $x_{[a,b]},x_{[c,d]}=\epsilon$, we replace them by $[a,d],[b,c]$ with $x_{[a,d]},x_{[b,c]}=\epsilon$. It is easily verified that this process doesn't affect the feasibility of the solution. This process is repeated until no such pair of intervals remain in the instance. We therefore assume from now on, that the intervals in $F$ are non-crossing. We would like to stress that the above is done only for the ease of analysis and during the course of the proof, it will be clear that we don't actually need to do this.

We partition the time slots in $[0,D]$ into \emph{blocks} and \emph{non-blocks} as follows. A duration $[a,b]$ is called a block iff $m_t \in [0,1)$ for all $a \le t \le b-1$ and $m_{a-1},m_b \ge 1$. A duration $[a,b]$ is called a non-block iff $m_t \in [1,m]$ for all $a \le t \le b-1$ and $m_{a-1},m_b < 1$. The following lemma leverages the new constraints and lower bounds the total weight of intervals of $F$ contained in a non-block. Proof of this lemma crucially uses the fact that $F$ is an optimum solution maximizing the value $\sum_{t} \min(m_t,1)$. 
\begin{lemma} \label{lemma:intervals_non_block}
Let $N=[a,b]$ be a non-block. If there exists a $t \in [a,b-1]$ such that $m_t>1$, then $\sum_{I:I \cap [a,b] \neq \emptyset}x_I \geq 2$. 
\end{lemma}
\begin{proof}
By noting that a non-block is a maximal contiguous set of time intervals with $m_t\geq 1$, intervals in $F$ are laminar and $m_t >1$ for some $t\in [a,b]$, there must exist an $I \in F$ such that $I \subseteq [a,b]$. Let $F'$ be the solution obtained by replacing $I=[l,r]$ by $I'=[l+1,r]$ in $F$. Since $F$ is a fractional solution with minimum cost, $F'$ must be infeasible. Let $\mathcal{D} = \{D_1,D_2,\dots D_\ell\}$ be the \emph{minimal} set of disjoint intervals of maximum deficiency (with respect to $\mathcal{I}$) as returned by Theorem \ref{thm:feasibility} and $J'$ be the set of jobs such that $fv(j_i,\mathcal{D})>0$ for all $j_i \in J'$. Note that the deficiency of $\mathcal{D}$ with respect to the current solution is $\epsilon$. Let $m'_t$ be defined with respect to $F'$.

Suppose there exists a $t$ such that $[t,t+1] \subseteq D\in \mathcal{D}$, $[t-1,t]$ is part of a non-block and $[t,t+1]$ is part of a block (or vice versa). This implies that $m'_t < 1$ and $m'_{t-1} \ge 1$. Since $m'_{t-1} > m'_t$, there exists an $I' \in F'$ which ends at $t$. If we extend $I'$ to the right by 1 unit, deficiency would decrease by $\epsilon$ (by Theorem \ref{thm:feasibility}) and we will obtain a feasible solution different from $F$, with greater value of $\sum_{t}\min(m_t,1)$ (as extending $I'$ implies setting $m'_{t}=m_{t}+\epsilon > m_t$). Hence, no $D \in \mathcal{D}$ overlaps with a block and non-block simultaneously. Thus we can partition the intervals in $\mathcal{D}$ depending on whether they are contained in a block or a non-block. Let $\mathcal{D}_{NB} \subseteq \mathcal{D}$ and $\mathcal{D}_B \subseteq \mathcal{D}$ be the intervals of $\mathcal{D}$ contained in blocks and non-blocks respectively. 

Let $D_1,D_2 \in \mathcal{D}$ be such that $D_1,D_2$ do not belong to the same block or the same non-block. Suppose there exists a $j_i \in J'$ such that $[r_i,d_i] \cap D_1 \neq \emptyset$ and $[r_i,d_i] \cap D_2 \neq \emptyset$. Then there must exist a $t$ such that $[t,t+1] \subseteq [r_i,d_i]$ and $m_t <1$. Since $fv(j_i,\mathcal{D}\cup [t,t+1])=fv(j_i,\mathcal{D})+1$ and $m_t < 1$, we have that the deficiency of $\mathcal{D} \cup [t,t+1]$ is strictly more than the deficiency of $\mathcal{D}$. This contradicts the fact that $\mathcal{D}$ has maximum deficiency and hence no job in $J'$ overlaps with distinct $D_i,D_j$. Therefore jobs in $J'$ can be partitioned according to the block or non-block they overlap with. Let $J'_{N} \subseteq J'$ be the set of jobs with positive forced volume overlapping with the non-block $N$ and $\mathcal{D}_N \subseteq \mathcal{D}_{NB}$ be the set of intervals of $\mathcal{D}$ overlapping with $N$. 

Observe that $\mathcal{D}_N$ must contain the time slot $[l,l+1]$ and hence is non-empty (recall that $[l,r]$ was replaced by $[l+1,r]$ in $F$ to obtain $F'$). Also, $J'_{N} \neq \emptyset$, otherwise $\mathcal{D} \setminus \mathcal{D}_N$ would have been the minimal set with maximum deficiency. Hence, $def(J'_N,\mathcal{D}_N,F')>0$. Let $x$ be the left end point of leftmost interval in $\mathcal{D}_N$ and $y$ be the right end point of the rightmost interval in $\mathcal{D}_N$. Since $m'_t \geq 1$ for all $[t,t+1] \in [x,y] \subseteq N$ and $def(J'_N,\mathcal{D}_N,F)>0$, there must exist a time slot $t\in [x,y-1]$ such that $m_t \geq 2$ in any integer feasible solution. This implies that $m-l(x,y) \geq 2$ and hence $\sum_{I:I \cap [a,b] \neq \emptyset}x_I \geq m-l(a,b) \geq 2$.
\end{proof}

\paragraph*{Rounding Scheme.}
We next convert/round $F$ into an integral solution in a series of steps. We will denote the three intermediate solutions by $F_1,F_2,F_3$ and the corresponding $m_t's$ as $m^{1}_{t},m^{2}_{t},m^3_t$. Let $T_N,T_B$ be the set of time slots in non-blocks and blocks respectively. Let $P_F$ be the total available capacity in $F$ (ie. $P=\sum_{t}m_{t}$) and $P_{B},P_{N}$ be the total available capacity in the blocks and non-blocks respectively. Let $Q_F$ be the total wake up cost incurred by $F$, ie. $Q_F=q\sum_{I}x_{I}$. Then $Cost(F)=P_F+Q_F$ and $P_F=P_B+P_N$. For each non-block $N=[a,b]$ such that $m_t > 1$ for some $t \in [a,b-1]$, we add an additional supply interval $[a,b]$ with $x_{[a,b]}=1$ to $F$ and call this solution $F_1$.
\begin{claim}
$Cost(F_1) \leq Cost(F)+P_N+Q_F$.
\end{claim}
\begin{proof}
$F_1$ is constructed by adding an interval $[a,b]$ with $x_{[a,b]}=1$ for a non-block $[a,b]$ if $m_t>1$ for some $t \in [a,b]$. For each time slot $t$ in a non-block, we have $m_t \geq 1$ and hence the total length of new intervals added is at most $\sum_{[t,t+1] \subseteq T_N} m_t \leq P_N$. If we add an additional interval for a non-block $[a,b]$, then $\sum_{I:I \cap [a,b] \neq \emptyset}x_I \geq 2$ (by Lemma \ref{lemma:intervals_non_block}). Since the intervals in $F$ are non-crossing, the above implies that the total weight of intervals which are completely contained in $[a,b]$ is at least 1. Thus the wakeup cost of each new interval can be charged to the wakeup cost of intervals completely contained inside the corresponding non-block, and the total additional wake up cost incurred is no more than $Q_F$.  
\end{proof}
We convert $F_1$ into $F_2$ by deleting the portions of existing intervals in $F$ such that for each time slot $t\in T_N$, we have the property $m^{2}_{t}= \lfloor m^{1}_t \rfloor$. This operation might increase the total wake up cost, but since the processing cost gets decreased by at least as much, the overall cost of the solution does not increase. This allows us to state the following.
\begin{claim}\label{clm:cost_F_2}
$Cost(F_2) \leq Cost(F) +P_N+Q_F$. Also, $m^{2}_t$ is an integer and $m^{2}_t \geq m_t$ for each time slot $t \in T_N$.
\end{claim}
We now describe our third transformation. As discussed earlier, we may again assume that all the intervals in $F_2$ have a weight of exactly $\epsilon$ and are non-crossing. Consider the single machine instance $J_S=J_N \cup J_B$, where $J_B=\{j_i|j_i \in J, [r_i,d_i]\subseteq T_B\}$ consists of jobs in $J$ which are completely contained inside some block and $J_N$ consists of additional jobs defined as follows: for each time slot $t \in T_N$, there is a job $j_{t}$ with release time $t$, deadline $t+1$ and a processing requirement of 1. We now pick a subset $F'_2$ of intervals in $F_2$ which form a feasible solution to the LP relaxation of the minimum cost skeleton given in Section \ref{section:LP_relaxation_skeleton} of the appendix. $F'_2 \subseteq F_2$ is the set of intervals of maximum total length such that the total weight of intervals containing any particular time slot is at most 1. It is worth noting that any interval containing a time slot of some block is a part of $F'_2$.

\begin{claim}\label{clm:skeleton_feasible_2_appx}
$F'_2$ is a feasible fractional skeleton for $J_S$. 
\end{claim}
\begin{proof}
We first note that $F$ is a feasible fractional skeleton for the set of jobs $J_B$. We do not modify portions of intervals overlapping with a block in $F_1,F_2$ and $F'_2$ contains all intervals of $F_2$ overlapping with a block. Hence, $F'_2$ is a feasible fractional skeleton for $J_B$.

Let $N$ be  non-block and $[a,a+1] \subseteq N$. By definition of a non-block, $m^{2}_t\geq m_t \geq 1$ for all $[t,t+1] \in N$. Since the intervals in $F_2$ are non-crossing (or laminar), it must be true that there are at least $1/\epsilon$ intervals which contain $[a,b]$. By definition of $F'_2$, at least $1/\epsilon$ of them are also contained in $F'_2$. Hence, total weight of intervals in $F'_2$ overlapping with $[a,a+1]$ is at least 1. The total weight of intervals in $F'_2$ containing any particular time slot is at most 1 (by definition of $F'_2$), hence $F'_2$ is a feasible fractional skeleton.
\end{proof}
In Theorem \ref{theorem:integral-skelton} (see Section \ref{section:LP_relaxation_skeleton} of appendix), we show that the LP relaxation for the minimum cost skeleton is exact. Hence, there exists an integer skeleton $J_S$ of cost no more than $F'_2$. Then our solution $F_3$ is $(F \setminus F'_2) \cup S$. Observe that $Cost(F_3) \leq Cost(F_2)$ and $m^{3}_t$ is an integer for every time slot $t$. We may therefore assume that $x_I=1$ for each $I \in F_3$. 
We now describe the final phase our algorithm, where we convert $F_3$ into a feasible solution by extending some existing intervals using $EXT$-$ALG$ (see Section \ref{sec:multi-converting}). If $F_3$ is a feasible solution, our algorithm terminates. Otherwise we find a disjoint minimal set of intervals of maximum deficiency $\mathcal{D}=\{D_1,D_2,\ldots,D_k\}$ guaranteed by Theorem \ref{thm:feasibility}. In each subsequent iteration, we use $EXT$-$ALG$ to extend an interval of $F_3$ by 1 unit, thereby reducing the maximum deficiency by 1. Claim \ref{clm:can_always_extend} shows that if the current solution is infeasible and it is not possible to extend an interval to reduce the deficiency, then the original instance is infeasible. Hence the extension phase of the algorithm terminates with a feasible solution.
\begin{claim} \label{clm:can_always_extend}
Let $F_{curr}$ be the current solution. If $F_{curr}$ is infeasible and for all $I \in F_{curr}$ the following is true: if $I \cap D_i \neq \varnothing$, then $D_i \subseteq I$, then the original instance is infeasible. 
\end{claim}
\begin{proof}
Suppose there exists a $D_i=[a,b]$ such that $D_i \cap I=\emptyset$ for all $I \in F_{curr}$. Since $F_3$ is a skeleton for $J$, there is no job $j_k$ such that $[r_k,d_k] \subseteq D_i$. Hence, there must exist a job $j_k$ such that $fv(j_k,\mathcal{D})>0$, $D_i \cap [r_k,j_k] \neq \emptyset$ and $[r_k,d_k] \subsetneq D_i$ (otherwise $\mathcal{D}\setminus D_i$ would be the minimal set with maximum deficiency). In this case, one of $\mathcal{D} \cup [a-1,a]$ or $\mathcal{D} \cup [b,b+1]$ would have strictly more deficiency than $\mathcal{D}$, thus contradicting the fact that $\mathcal{D}$ has maximum deficiency. Hence, no such $D_i$ exists. 

From the discussion above, we have that for all $[t,t+1] \in \mathcal{D}$, $m^{curr}_{t} \geq 1$. Hence, $m^{curr}_{t} \geq m_t$ for all $[t,t+1] \subseteq T_B \cap \mathcal{D}$. By Claim \ref{clm:cost_F_2}, we have $m^{curr}_{t} \geq m^{2}_{t} \geq m_t$ for all $[t,t+1] \subseteq T_{N} \cap \mathcal{D}$. Hence the total processing available in $F_{curr}$ in $\mathcal{D}$ is at least the total processing available in $F$ in $\mathcal{D}$. This implies that $F$ is a infeasible fractional solution and hence the original instance is infeasible as well.
\end{proof}

The deficiency at the start of the extension phase can be at most $P_B$ as $m^{3}_{t} \geq m_t$ for $t \in T_N$. Since we decrease the deficiency by 1 in each iteration, there can be at most $P_B$ iterations of the extension phase. In each step we increase the cost of the solution by 1, hence cost of the final solution is at most $Cost(F_3)+P_B \leq Cost(F)+P_N+Q_F+P_B=Cost(F)+P+Q_F \leq 2 Cost(F)$. This shows that the integrality gap of the LP relaxation is at most $2$. To compute $F_1,F_2,F_3$, we only need the value of $m_t$'s and don't need to create multiple copies of intervals in our solution. Thus our rounding algorithm can be implemented in pseudo-polynomial time and this completes the proof of Theorem~\ref{thm:pseudo}.
\bibliographystyle{plainurl}
\bibliography{references}
\appendix
\section*{Appendix}
\section{Proof of Lemma \ref{lemma:extending_skeleton}}
\begin{proof}
We first use the well known Earliest Deadline First ($\tt{EDF}$) algorithm to compute the maximum volume of jobs that can be processed in $S$. EDF processes in each active time slot $[t,t+1]$, among the unfinished jobs whose span includes $[t,t+1]$ the one that has the earliest deadline (or none if no such job exists). Let this schedule be $\tt{Sch_1}$. Two important folklore properties of EDF that will be useful in this proof are the following. First, EDF can be used to test feasibility of the active intervals. In particular, if there exists a feasible schedule for a given instance, then one can by repeating a simple exchange argument argue that the EDF schedule is feasible as well. Secondly, by maintaining a balanced search tree of the  unfinished jobs whose span contains the current time-slot,  $\tt{Sch_1}$ can be computed in $O(n \log n)$ time. 
Let $p'_{1},p'_{2},\ldots,p'_{n}$ be the volume of respective jobs that is be processed in $\tt{Sch_1}$. Notice that $\sum_{i=1}^{n}p_i'=P_S$. We first remove all time slots from $S$ which were active in $S$ but not used for the processing of any job in $\tt{Sch_1}$ (as no job will be scheduled in these time slots during the course of the algorithm). We now describe an algorithm to make $S$ feasible by increasing the number of active time slots, without increasing the number of active intervals. We maintain a balanced binary search tree of the currently active intervals and maximal gaps (there is a natural total order on active intervals and maximal gaps) at every step of the algorithm.\\

The algorithm works in iterations. $S_i$ will denote the solution after iteration $i$ and $S_0=S$. In iteration $i$, we increase the number of active time slots in $S_{i-1}$ such that $\{p_1,\ldots,p_i,p'_{i+1},\ldots,p'_n\}$ volume of jobs can be feasibly scheduled in $S_i$. We first find the active interval or maximal gap containing $d_i$. If $d_i$ is contained in an active interval, we make the first $p_i-p_i'$ inactive time slots to the left of $d_i$ active. We don't create any new active intervals during the process but may have to set some currently inactive time slots to the left of $r_i$ as active. This can be done efficiently using the balanced binary search tree on maximal gaps/active intervals and may require merging of active intervals. If $d_i$ is contained in a maximal gap $g=[a,b]$, we increase the length of the active interval ending at $a$. If $d_i-a \geq p_i-p_i'$, we set all the inactive time slots in $[a,a+p_i-p_i']$ as active and move on to the next iteration. Since $S$ is a skeleton, $a \geq r_i$, and we can already process $\{p_1,\ldots,p_i,p'_{i+1},\ldots,p'_n\}$ volume of jobs in this case. If $d_i-a \leq g$, we first set all the active time slots in $[a,d_i]$ as active and then set the first $p_i-p_i'-|d_i-a|$ inactive slots to the left of $d_i$ as active.\\

In iteration $i$, we increase the length of active intervals (without creating any new ones) by $p_i-p_i'$. Hence the total increase in the cost of the solution over all iterations is $P-P_S$. In iteration $i$, if we extend an active interval in $k_i$ maximal gaps of $S_{i-1}$, then we completely fill at least $k_i-1$ of them. Since there are at most $n$ maximal gaps at the beginning, we have $\sum_{i=1}^{n}(k_i-1) \leq n$. This implies $\sum_{i=1}^{n}k_i \leq 2n$ and at most $2n$ maximal gaps (or active intervals) are modified over the course of the entire algorithm. Using a balanced binary search tree, time required to find and operate on a maximal gap in each iteration is $O(\log n)$. Since the total number of times we operate on maximal gaps is at most $2n$, the whole algorithm can be implemented in $O(n\log n)$ time. We now show that the solution returned at the end of the algorithm is a feasible solution to the original instance.\\

Using induction, we show that for each $i \in [1,n]$, $\{p_1,\ldots,p_i,p'_{i+1},\ldots,p'_n\}$ volume of respective jobs can be feasibly scheduled in $S_i$. For $i=0$, the induction statement is vacuously true. Assume that the above is true for $S_{i-1}$. If $S_i$ doesn't satisfy the above condition, then there exists a $[x,y], x \leq r_i, y\geq d_i$ such that the total number of active time slots in $[x,y]$ (say $V$) is less than the sum of processing times of all jobs with both release and deadline inside $[x,y]$ (say $F$). \\

If there is a maximal gap overlapping with $[x,d_i]$ in $S_i$, then we must have added at least $p_i-p_i'$ new active time slots in $[x,d_i]$ in iteration $i$. Let $V'$ be the number of active time slots in $[x,y]$ in $S_{i-1}$.  Since $\{p_1,\ldots,p_{i-1},p'_{i},\ldots,p'_n\}$ volume of jobs can be feasibly scheduled in in $S_{i-1}$, $V' \geq F-(p_i-p_i')$. Since we add at least $p_i-p_i'$ active time slots in $[x,y]$ in iteration $i$, we have $V=V'+(p_i-p_i') \geq F$, a contradiction. \\

Suppose all the time slots in $[x,d_i]$ are active in $S_i$. We use $\tt{EDF}$ to find the maximum volume that can be processed in $S_{i}$ for the instance with processing times $\{p_1,\ldots,p_i,p'_{i+1},\ldots,p'_n\}$ (we do this only for the purpose of analysis). Let the scheduled returned by $\tt{EDF}$ be $\tt{Sch_2}$. If $\tt{Sch_2}$ doesn't schedule any job with deadline $>d_i$ in $[x,d_i]$, then $\tt{Sch_1}$ doesn't schedule any jobs with deadline $>d_i$ in $[x,d_i]$ as well. This implies that the number of active slots available in $[d_i,y]$ is at least $\sum_{k:d_k >d_i,[r_k,d_k]\subseteq [x,y]}p'_{k}$. Therefore it must be true that the number of active time slots in $[x,d_i]$ is strictly less than $\sum_{k:[r_k,d_k] \subseteq [x,d_i]}p_i$. This implies that the original instance is infeasible, a contradiction. \\

Suppose $\tt{Sch_2}$ schedules a job with deadline after $d_i$ in $[x,d_i]$. Let $x',x<x'<d_i$ be the right most time slot in which a job with deadline $>d_i$ is scheduled. In $\tt{Sch_1,Sch_2}$, all jobs scheduled in $[x'+1,d_i]$ must have both their release and deadline in $[x'+1,d]$. Since the original instance is feasible and all the time slots in $[x'+1,d_i]$ are active, all jobs having both their release and deadline in $[x'+1,d]$ are feasibly scheduled in $\tt{Sch_2}$. Since we only increase the number of active time slots in going from $S_{i-1}$ to $S_{i}$, rest of the jobs are also feasibly scheduled in $\tt{Sch_2}$ and hence $S_i$. Hence, $\{p_1,\ldots,p_i,p'_{i+1},\ldots,p'_n\}$ volume of jobs can be processed in $S_i$ and this completes the proof of the invariant. Since $S_n$ corresponds exactly to the original instance, the solution constructed by the algorithm is feasible. \end{proof}

\section{A Linear Programming Relaxation for Skeletons} \label{section:LP_relaxation_skeleton}

In this section, we give a linear programming relaxation for computing the minimum cost skeleton for single machine and then go on to show that this relaxation is exact, ie. there exists an feasible (integer) skeleton with cost at most the optimum value of the linear program. We use this fact to design a 2-approximation algorithm for the multiple processor case. The main aim of this section is to show that the linear program has an integer optimal solution, one can always use the dynamic program approach described before to compute the optimal solution. The linear program and the proof of its integrality are essentially present in the work of ~\cite{AGKK20}.

 \vspace{0.1in}
 \[	 \min \quad  \sum_{I} x_I (q + |I|)\]
 	
 	\[ \quad \displaystyle\sum_{I : t \in I  } x_I \le 1  \quad \forall t \in [0,D]\] 
 	\[ \displaystyle\sum_{I : I \cap [r_i,d_i] \neq \emptyset  } x_I \ge 1   \forall i \in [1,n] \]
 	\[  \quad x_I \geq 0 \quad \quad \quad  \forall I\subseteq [a,b] \]
 
 \vspace{0.1in}

\begin{theorem} \label{theorem:integral-skelton}
The linear programming relaxation for the skeleton problem has an optimal integer solution.
\end{theorem}
\begin{proof}

Let $x$ be the optimum fractional solution and let $F=\set{I|x_I>0}$. Let $\epsilon = \gcd_{i \in F}(x_I)$. We create $x_I / \epsilon$ copies of each $I \in F$ to assume that all intervals in $F$ have the same $x_I$ value. If there exist $[a,b],[c,d], a<c<d<b$ with $x_{[a,b]},x_{[c,d]}=\epsilon$, we replace them by $[a,d],[c,b]$ with $x_{[a,d]},x_{[c,b]}=\epsilon$. We repeat this process until no such pair of intervals remain and may therefore assume for the remainder of this proof that the all intervals in $F$ are crossing. This allows us to order the intervals in $F$ from left to right using their start time, say $I_1 \prec I_2 \prec \dots$. 
We construct $t=1/\epsilon$ integer solutions  $F_1,F_2,\ldots,F_{t}$ as follows: the intervals in $F$ are assigned to the solutions $F_1,F_2,\ldots,F_{t}$ in a round robin manner, ie. the solution $F_j$, $1 \le j \le t$ consists of the intervals $I_j, I_{j+t},I_{j+2t},\dots$. Note that the solution constructed are integer solutions, ie. $x_I=1$ for all $I \in F_1,F_2,\ldots$.

We now show that $F_1,F_2,\ldots,F_{t}$ are feasible skeletons. This will show that the average cost of $F_1,F_2,\ldots$ is equal to $Cost(F)$ and hence there must exist $ i \in [ 1, t]$ such that $Cost(F_i) \leq Cost(F)$. We first show that no two intervals in any $F_j$, $1 \le j \le t$ overlap with each other.
\begin{claim}\label{cl:disjoint}
All the intervals in $F_j$ are pairwise disjoint for any $1 \le j \le t$.
\end{claim}
\begin{proof}
Let $I_k \prec I_{k+lt} \in F_j$ overlap with each other for some $k$ and $l \ge 1$. Suppose $t \in  I_k \cap I_{k+lt}$. Therefore, any interval $I$ satisfying $I_k \prec I \prec I_{k+lt}$ contains $t$. Therefore we have $\sum_{I: I_k \preceq I \preceq I_{k+lt}}x_I = (lt+1) \epsilon >1$ which violates the LP-constraint $\sum_{I:t\in I} x_I\leq 1$ and yields a contradiction.
\end{proof}
\begin{claim}\label{cl:skelton}
For any job $j_i$ and $1 \le j \le t$, we have $I \cap [r_i,d_i] \neq \emptyset$ for some $I \in F_j$.
\end{claim}
\begin{proof}
Suppose none of the intervals of $F_j$ overlaps with $[r_i,d_i]$. Then at most $(t-1)$ intervals of $F$ can overlap with $[r_i,d_i]$. Thus $\sum_{I : I \cap [r_i,d_i] \neq \emptyset  } x_I \le (t-1)\epsilon < 1$ which violates our second LP constraint.
\end{proof}
Hence $F_1,\dots,F_t$ are feasible skeletons and each is an integer solution for the problem. Since their average cost is at most that of $F$, there exists a $F_i$ with cost exactly equal to that of $F$. This completes the proof of the theorem.
\end{proof}

\end{document}